\documentclass[11pt,hidelinks,a4paper]{article}

\usepackage[utf8]{inputenc}
\usepackage[english]{babel}
\usepackage[T1]{fontenc}

\usepackage{mathtools}
\usepackage{amsmath, amsthm, amssymb,amsfonts}
\usepackage{leftindex}
\usepackage{thmtools,thm-restate}
\usepackage{proof-at-the-end}
\PassOptionsToPackage{hidelinks}{hyperref}
\usepackage[hidelinks]{hyperref}
\usepackage[nameinlink,noabbrev]{cleveref}
\usepackage{placeins}
\usepackage{graphicx}
\usepackage{float}
\usepackage{subcaption}
\usepackage[skip=2pt]{caption}

\usepackage{dsfont}
\usepackage{physics}
\usepackage{mleftright} %

\usepackage{booktabs}

\usepackage{natbib} 
\usepackage[a4paper, total={16cm, 22cm}]{geometry}

\newtheorem{assumption}{Assumption}
\newtheorem{model}{Model}
\newtheorem{theorem}{Theorem}

\newtheorem{lemma}[theorem]{Lemma}

\newcommand{\BP}{{\mathbb{P}}}

\newcommand{\BR}{{\mathbb{R}}}

\newcommand{\CN}{{\cal N}}
\newcommand{\CO}{{\cal O}}

\DeclareMathOperator{\Cov}{Cov}

\DeclareMathOperator{\LR}{LR}

\DeclareMathSymbol{\shortminus}{\mathbin}{AMSa}{"39}

\DeclareMathOperator{\diag}{diag}

\DeclareMathOperator*{\plim}{plim}
\DeclareMathOperator{\vecop}{vec}

\DeclarePairedDelimiterX{\infdivx}[2]{(}{)}{%
  #1\;\delimsize\|\;#2%
}

\newcommand{\Id}{\mathrm{Id}}
\newcommand{\tod}{\overset{d}{\to}}
\newcommand{\toP}{\overset{\BP}{\to}}

\let\phi\varphi
\let\theta\vartheta
\let\epsilon\varepsilon

\let\leq\leqslant
\let\geq\geqslant

\newcommand\numberthis{\addtocounter{equation}{1}\tag{\theequation}}

\title{
    The exact distribution of the conditional likelihood-ratio test in instrumental variables regression
}
\author{
Malte Londschien
\\
\vspace{0cm}\\
{\small Seminar for Statistics, ETH Z\"urich, Switzerland}\\
{\small AI Center, ETH Z\"urich, Switzerland}\\
}

\date{September 2025}
\begin{document}
\maketitle

\begin{abstract}
    \noindent We derive the exact asymptotic distribution of the conditional likelihood-ratio test in instrumental variables regression under weak instrument asymptotics and for multiple endogenous variables.
    The distribution is conditional on all eigenvalues of the concentration matrix, rather than only the smallest eigenvalue as in an existing asymptotic upper bound.
    This exact characterization leads to a substantially more powerful test if there are differently identified endogenous variables.
    We provide computational methods implementing the test and demonstrate the power gains through numerical analysis.
\end{abstract}

\section{Introduction}
Instrumental variables regression allows for the estimation of causal effects in the presence of unobserved confounding by exploiting variation in treatment variables induced by so-called instruments, variables that affect the outcome only through the treatment.
In practice, to make evidence-based policy decisions, reliable uncertainty quantification is essential.
Standard methods to construct $p$-values and confidence sets rely on the asymptotic normality of estimators such as the two-stage least squares or limited information maximum likelihood estimators.

\citet{staiger1997instrumental} show that when the instruments are weak, as is common in empirical economics,
these tests have incorrect size and reject the null hypothesis too often. %
To study this phenomenon, \citet{staiger1997instrumental} propose weak-instrument-asymptotics, a theoretical framework where instrument strength decreases as the number of samples increases, and the first-stage F-statistic is of constant order.
Several weak-instrument-robust tests exist that have the correct size under weak-instrument-asymptotics.
These include the Anderson-Rubin test \citep{anderson1951estimating}, the Lagrange multiplier test \citep{kleibergen2002pivotal}, and the conditional likelihood-ratio test \citep{moreira2003conditional}.

For a single endogenous variable, \citet{moreira2003conditional} derives the asymptotic distribution of the likelihood-ratio test statistic, conditional on the concentration parameter.
The resulting conditional test has correct size even if instruments are weak.
Given multiple endogenous variables, \citet{kleibergen2007generalizing} provides an asymptotic upper bound of the test's distribution, conditional on the smallest eigenvalue of the concentration matrix.
If all eigenvalues of the matrix are equal, this bound is sharp.

We compute the exact asymptotic distribution of the conditional likelihood-ratio test for multiple endogenous variables under weak-instrument-asymptotics.
This distribution is conditional on all eigenvalues of the concentration matrix rather than just the smallest.
This exact characterization substantially improves power when instruments vary in strength across endogenous variables or the endogenous variables are correlated, a common scenario that leads to differing eigenvalues of the concentration matrix.

We propose computation methods for the test's critical values and analyse its power in numerical analyses.
The test is implemented in the Python package \texttt{ivmodels} \citep{londschien2024weak,londschien2025statistician}.

\section{Main result}

We consider a standard instrumental variables regression model with weak instruments.
\begin{model}
    \label{model:0}
    Let $y_i = X_i^T \beta_0 + \varepsilon_i \in \BR$ with $X_i = Z_i^T \Pi + V_{X, i} \in \BR^m$ for random vectors $Z_i \in \BR^k, V_{X, i}\in \BR^m$, and $\varepsilon_i \in \BR$ for $i=1\ldots, n$ and parameters $\Pi \in \BR^{k \times m}$, and $\beta_0 \in \BR^m$.
    The $Z_i$ are \emph{instruments}, the $X_i$ are \emph{endogenous covariates}, and the $y_i$ are \emph{outcomes}.
    We consider \emph{weak instrument asymptotics} \citep{staiger1997instrumental}, where $\sqrt{n} \Pi = \Pi_0$ is fixed and of full column rank $m$ and thus $\Pi = \CO(\frac{1}{\sqrt{n}})$.
\end{model}%
\noindent Assume that a central limit theorem applies to the sums $Z^T \varepsilon$ and $Z^T V_X$.
\begin{assumption}
    \label{ass:0}
    Let
    $$
    \Psi := \begin{pmatrix} \Psi_{\varepsilon} & \Psi_{V_X} \end{pmatrix} := (Z^T Z)^{-1/2} Z^T \begin{pmatrix} \varepsilon & V_X \end{pmatrix} \in \BR^{k \times (1 + m)}.
    $$
    Assume there exist $\Omega \in \BR^{(1+m) \times (1+m)}$ and $Q \in \BR^{k \times k}$ positive definite such that, as $n \to \infty$,
    \begin{align*}
        &\mathrm{(a)} \ \ \frac{1}{n} \begin{pmatrix}\varepsilon & V_X \end{pmatrix}^T \begin{pmatrix}\varepsilon & V_X \end{pmatrix} \toP \Omega = \begin{pmatrix}
    \sigma^2_\varepsilon & \Omega_{\varepsilon, V_X} \\
    \Omega_{V_X, \varepsilon} & \Omega_{V_X} \\
\end{pmatrix}, \\
        &\mathrm{(b)} \ \ \vecop(\Psi) \tod \CN(0, \Omega \otimes \Id_k), \text{ and }\\
        &\mathrm{(c)} \ \ \frac{1}{n} Z^T Z \toP Q,
    \end{align*}
    where $\Cov(\vecop(\Psi)) = \Omega \otimes \Id_k$ means $\Cov(\Psi_{i, j}, \Psi_{i', j'}) = 1_{i = i'} \cdot \Omega_{j, j'}$.
\end{assumption}
\noindent \Cref{ass:0} is similar to the assumptions of \citeauthor{moreira2003conditional}'s \citeyearpar{moreira2003conditional} theorem 2 and is a special case of \citeauthor{kleibergen2007generalizing}'s \citeyearpar{kleibergen2007generalizing} assumption 1.
\citet{londschien2025statistician} show in their Lemma 1 that if the $(Z_i, \varepsilon_i, V_{X, i})$ are i.i.d.\ with finite second moments and conditional homoscedasticity, then \cref{ass:0} holds.

Denote with $P_Z = Z (Z^T Z)^{-1} Z^T$ the projection matrix onto the space spanned by $Z$ and $M_Z = \Id_n - P_Z$ the projection onto the orthogonal complement.

\begin{theoremEnd}[one big link translated={\hspace{-0.71cm} See proof on page}]{theorem}%
    \label{theorem:1}
    Assume \cref{model:0} and \cref{ass:0} holds.
    Let
    $$
    \LR(\beta) := (n - k) \frac{ (y - X \beta)^T P_Z (y - X \beta) }{ (y - X \beta)^T M_Z (y - X \beta) } - (n - k) \, \min_b  \frac{ (y - X b)^T P_Z (y - X b) }{ (y - X b)^T M_Z (y - X b) }
    $$
    be the likelihood-ratio test for the causal parameter $\beta$ in instrumental variables regression.
    Let $\tilde X(\beta) := X - (y - X \beta) \frac{ (y - X \beta)^T M_Z X }{ (y - X \beta)^T M_Z (y - X \beta) }$ and let $\lambda_1(\beta), \ldots, \lambda_m(\beta)$ be the eigenvalues of the matrix $(n - k) \, \left[ \tilde X(\beta)^T M_Z \tilde X(\beta) \right]^{-1} \tilde X(\beta)^T P_Z \tilde X(\beta)$.
    Let $q_0 \sim \chi^2(k - m)$ and $q_1, \ldots, q_m \sim \chi^2(1)$ be independent of each other.
    Denote with $\mu_\mathrm{min}(\lambda_1, \ldots, \lambda_m, q_0 \ldots, q_m)$ the smallest root of the polynomial
    $$
    p_{\lambda_1, \ldots, \lambda_m, q_0 \ldots q_m}(\mu) := \left(\mu - \sum_{i = 0}^m q_i \right) \cdot \prod_{i = 1}^m (\mu - \lambda_i) - \sum_{i = 1}^m \lambda_i q_i \prod_{j \geq 1, j\neq i} (\mu - \lambda_j).
    $$
    This satisfies $0 \leq \mu_\mathrm{min}(\lambda_1, \ldots, \lambda_m, q_0 \ldots, q_m)\leq \min(\lambda_1, q_0)$ and, conditionally on $\lambda_1(\beta_0), \ldots, \lambda_m(\beta_0)$,
    $$
    \LR(\beta_0) \overset{d}{\to} \sum_{i = 0}^m q_i - \mu_\mathrm{min}(\lambda_1(\beta_0), \ldots, \lambda_m(\beta_0), q_0, \ldots, q_m) \text{ as } n \to \infty.
    $$
\end{theoremEnd}%
\begin{proofEnd}%
    By Corollary 9 and Proposition 10 of \citet{londschien2025statistician}, we have that
    \begin{align*}
    \lambda &:=  (n - k) \, \min_b \frac{ (y - X b)^T P_Z (y - X b) }{ (y - X b)^T M_Z (y - X b) } \\
    &=  \lambda_\mathrm{min} \left( (n - k) \left[ \begin{pmatrix} y & X \end{pmatrix}^T M_Z \begin{pmatrix} y & X \end{pmatrix} \right]^{-1} \begin{pmatrix} y & X \end{pmatrix}^T P_Z \begin{pmatrix} y & X \end{pmatrix} \right)
    \end{align*}
    Write $\tilde X := \tilde X (\beta_0)$ and $\lambda_1, \ldots, \lambda_m = \lambda_1(\beta_0), \ldots, \lambda_m(\beta_0)$.
    Calculate
    \begin{equation} \label{eq:clr1}
        \begin{pmatrix} y & X \end{pmatrix} \begin{pmatrix} 1 & 0 \\ - \beta_0 & \Id_m \end{pmatrix} 
        \begin{pmatrix}
             1 & -\frac{\varepsilon^T M_Z X}{\varepsilon^T M_Z \varepsilon} \\ 0 & \Id_m
         \end{pmatrix}
        =
        \begin{pmatrix} \varepsilon & X \end{pmatrix} 
        \begin{pmatrix}
            1 & -\frac{\varepsilon^T M_Z X}{\varepsilon^T M_Z \varepsilon} \\ 0 & \Id_m
        \end{pmatrix} =
        \begin{pmatrix} \varepsilon & \tilde X \end{pmatrix}
    \end{equation}
    Note that $\varepsilon^T M_Z \tilde X = 0$ and thus
    \begin{equation}\label{eq:clr2}
        \widehat{\tilde \Omega} := \frac{1}{n-k}
        \begin{pmatrix} \varepsilon & \tilde X \end{pmatrix}^T M_Z \begin{pmatrix} \varepsilon & \tilde X \end{pmatrix} = 
        \frac{1}{n-k}
        \begin{pmatrix}
            \varepsilon^T M_Z \varepsilon & 0 \\
            0 & \tilde X^T M_Z \tilde X
        \end{pmatrix}        = : \begin{pmatrix}
        \hat\sigma^2 & 0 \\ 0 & \widehat{\tilde \Omega}_{V_X}
        \end{pmatrix}
    \end{equation}
    Calculate
    \begin{align*}
    \lambda %
    &= \min \{\mu \in \BR \mid \det( 
        \mu \cdot \Id_{m+1} - (n - k) \left[\begin{pmatrix} y & X \end{pmatrix}^T M_Z \begin{pmatrix} y & X \end{pmatrix}\right]^{-1} \begin{pmatrix} y & X \end{pmatrix}^T P_Z \begin{pmatrix} y & X \end{pmatrix}
    ) = 0 \}\\
    &= 
    \min \{ \mu \in \BR \mid \det( 
        \frac{\mu}{n-k} \cdot \begin{pmatrix} y & X \end{pmatrix}^T M_Z \begin{pmatrix} y & X \end{pmatrix} -  \begin{pmatrix} y & X \end{pmatrix}^T P_Z \begin{pmatrix} y & X \end{pmatrix}
    ) = 0\}\\
    &\overset{\text{(\ref{eq:clr1}, \ref{eq:clr2})}}{=} 
    \min \{ \mu \in \BR \mid \det( 
        \mu \cdot \Id_{m+1}
        - \widehat{\tilde \Omega}^{-1/2, T} \begin{pmatrix} \varepsilon & \tilde X \end{pmatrix}^T P_Z \begin{pmatrix} \varepsilon & \tilde X \end{pmatrix}  \widehat{\tilde \Omega}^{-1/2}
    ) = 0\}.
    \end{align*}

    Let $U D V =  (Z^T Z)^{-1/2} Z^T \tilde X \, {\widehat{\tilde \Omega}}_{V_X}^{-1/2}$ be a singular value decomposition with $D^2 = \diag( \lambda_1, \ldots, \lambda_m )$ containing the eigenvalues of $ (n - k) \cdot (\tilde X^T M_Z \tilde X)^{-1} \tilde X^T P_Z \tilde X$.
    Let $U_i$ be the $i$-th column of $U$ for $i = 1, \ldots, m$.
    Then $U_i^T U_j = 0$ for $i \neq j$ and $1$ otherwise.
    Calculate
    \begin{align*}
       \Sigma:=\widehat{ \tilde \Omega }^{-1/2, T} \begin{pmatrix} \varepsilon & \tilde X \end{pmatrix}^T &P_Z \begin{pmatrix} \varepsilon & \tilde X \end{pmatrix} \widehat{ \tilde \Omega }^{-1/2} = 
        \begin{pmatrix} 1 & 0 \\ 0 & V^T \end{pmatrix}
        \begin{pmatrix} \varepsilon^T P_Z \varepsilon / \hat\sigma^2 & \Psi_\varepsilon^T U D / \hat\sigma \\ D U^T \Psi_\varepsilon / \hat\sigma & D^2 \end{pmatrix}
        \begin{pmatrix} 1 & 0 \\ 0 & V \end{pmatrix}
    \end{align*}
    such that
    \begin{align*}
        \det \left( \mu \cdot \Id_{m + 1} - \Sigma \right) &= \det \begin{pmatrix} \mu - \varepsilon^T P_Z \varepsilon / \hat\sigma^2 & \Psi_\varepsilon^T U_1 \sqrt{\lambda_1} /  \hat\sigma & \cdots & \Psi_\varepsilon^T U_m \sqrt{\lambda_m} /  \hat\sigma \\
            \sqrt{\lambda_1} U_1^T \Psi_\varepsilon /  \hat\sigma & \mu - \lambda_1 & \cdots & 0 \\
        \vdots & \vdots& \ddots & \vdots \\
        \sqrt{\lambda_m} U_m^T \Psi_\varepsilon /  \hat\sigma & 0 & \cdots & \mu - \lambda_m
    \end{pmatrix}.
    \end{align*}
    We apply \cref{lem:1} with $d_0 = \mu - \varepsilon^T P_Z \varepsilon / \hat\sigma^2$, $d_i = \mu - \lambda_i$, and $a_i = \Psi_\varepsilon^T U_i \sqrt{\lambda_i} / \hat \sigma$ for $i = 1, \ldots, m$.
    Then
    $$
    \det \left( \mu \cdot \Id_{m + 1} - \Sigma \right) = (\mu - \varepsilon^T P_Z \varepsilon / \hat\sigma^2) \cdot \prod_{i = 1}^m (\mu - \lambda_i) - \sum_{i = 1}^m (\Psi_\varepsilon^T U_i)^2 \lambda_i / \hat\sigma^2 \prod_{j \geq 1, j \neq i} (\mu - \lambda_i) .
    $$ 
    Define $q_i := (\Psi_\varepsilon^T U_i)^2 / \hat\sigma^2 = \Psi_\varepsilon^T P_{U_i} \Psi_\varepsilon / \hat\sigma^2$ and $q_0 := \Psi_\varepsilon^T (\Id_k - P_U) \Psi_\varepsilon / \hat\sigma^2$.
    Then, $\det\left( \mu \cdot \Id_{m + 1} - \Sigma \right) = p(\mu)$ and $\lambda = \mu_\mathrm{min}(\lambda_1, \ldots, \lambda_m, q_0, \ldots, q_m)$.

    It remains to show that the $q_i \to_d \chi^2(1)$ for $i = 1, \ldots, m$ and $q_0 \to_d \chi^2(k - m)$, asymptotically independent of each other and of $(n - k) [\tilde X^T M_Z \tilde X ]^{-1/2} \tilde X^T P_Z \tilde X [\tilde X^T M_Z \tilde X ]^{-1/2} $.

    Write
    $$
    \Omega = \begin{pmatrix} \sigma^2_\varepsilon & \Omega_{\varepsilon, V_X} \\ \Omega_{V_X, \varepsilon} & \Omega_{V_X} \end{pmatrix} \text{ and } \tilde \Omega :=         \begin{pmatrix}
            1 & -\Omega_{\varepsilon, V_X} / \sigma^2_\varepsilon \\ 0 & \Id_m
        \end{pmatrix}^T \Omega         \begin{pmatrix}
            1 & -\Omega_{\varepsilon, V_X} / \sigma^2_\varepsilon \\ 0 & \Id_m
        \end{pmatrix} =: \begin{pmatrix}
            \sigma^2_\varepsilon & 0 \\
            0 & \tilde \Omega_{V_X}
        \end{pmatrix}.
    $$
    By \cref{ass:0} (a), we have $\widehat{\tilde \Omega} \to_\BP \tilde \Omega$.
    Define $\Psi_{\tilde X} := (Z^T Z)^{-1/2} Z^T \tilde X \to_\BP (Z^T Z)^{1/2} \Pi + \Psi_{V_X} - \Psi_{\varepsilon} \Omega_{\varepsilon, V_X} / \sigma^2_\varepsilon$ as $\frac{\varepsilon^T M_Z X}{\varepsilon^T M_Z \varepsilon} \to_\BP \Omega_{\varepsilon, V_X} / \sigma^2_\varepsilon$ by \cref{ass:0} (a).
    Then, by \cref{ass:0} (b, c) and as $\Pi = \frac{1}{\sqrt{n}} \Pi_0$:
    $$\vecop(\Psi_\varepsilon, \Psi_{\tilde X} ) \tod \CN \left( (0, Q^{1/2} \Pi_0), \tilde \Omega \otimes \Id_k \right).
    $$
    As the off-diagonal terms of $\tilde \Omega$ are zero, this implies that $\Psi_\varepsilon$ and $\Psi_{\tilde X}$ are asymptotically jointly Gaussian and asymptotically independent.
    Then, also $\Psi_\varepsilon$ and 
    $$\plim \ (n - k) [\tilde X^T M_Z \tilde X ]^{-1/2} \tilde X^T P_Z \tilde X [\tilde X^T M_Z \tilde X ]^{-1/2} = \tilde \Omega_{V_X}^{-1/2} \Psi_{\tilde X}^T \Psi_{\tilde X} \tilde \Omega_{V_X}^{-1/2}$$
    are asymptotically independent.
    
    We condition on $ (n - k) [\tilde X^T M_Z \tilde X ]^{-1/2} \tilde X^T P_Z \tilde X [\tilde X^T M_Z \tilde X ]^{-1/2}$ (with eigenvalues $\lambda_1, \ldots, \lambda_m$).
    We apply Cochran's theorem with $\Psi_\varepsilon / \sigma_\varepsilon^2 \sim \CN(0, \Id_k)$ and $A_i := P_{U_i} = U_i U_i^T$ for $i=1, \ldots m$ and $A_0 := \Id - U U^T = M_U = \Id_k - \sum_{i=1}^m A_i$ of ranks 1 and $k-m$.
    This yields that the $q_i = (\Psi_\varepsilon^T U_i)^2 / \hat\sigma^2 = \Psi_\varepsilon^T U_i U_i^T \Psi_\varepsilon / \hat\sigma^2 \to_\BP \Psi_\varepsilon^T A_i \Psi_\varepsilon / \sigma^2_\varepsilon \to_d \chi^2(1)$ independently and $q_0 = \Psi_\varepsilon^T A_0 \Psi_\varepsilon / \hat\sigma^2 \to_d \chi^2(k - m)$.
\end{proofEnd}
This directly implies \citeauthor{moreira2003conditional}'s \citeyearpar{moreira2003conditional} result for $m=1$ and \citeauthor{kleibergen2007generalizing}'s \citeyearpar{kleibergen2007generalizing} upper bound for $m \geq 1$.
\begin{theoremEnd}[one big link translated={\hspace{-0.71cm} See proof on page}]{corollary}[\citeauthor{moreira2003conditional}, \citeyear{moreira2003conditional}]
    \label{cor:1}
    If $m=1$, then, conditionally on 
    $$\lambda_1 := (n - k) \left[ \tilde X(\beta_0)^T M_Z \tilde X(\beta_0) \right]^{-1} \tilde X(\beta_0)^T P_Z \tilde X(\beta_0),$$
    we have
    $$
    \LR(\beta_0) \tod \Gamma(k - 1, 1, \lambda_1),
    $$
    where 
    $$\Gamma(k - 1, 1, \lambda_1) \overset{d}{:=} \frac{1}{2} (q_0 + q_1 - \lambda_1 + \sqrt{(q_0 + q_1 + \lambda_1)^2 - 4 q_0 \lambda_1})$$
    for $q_0 \sim \chi^2(k - 1)$ and $q_1 \sim \chi^2(1)$ independent.
\end{theoremEnd}
\begin{proofEnd}
    If $m=1$ then
    $$p(\mu) = (\mu - q_0 - q_1)(\mu - \lambda_1) - \lambda_1 q_1 = \mu^2 - (q_0 + q_1 + \lambda_1) \mu + q_0 \lambda_1.$$
    This has roots $\mu_\pm = \frac{1}{2} (q_0 + q_1 + \lambda_1 \pm \sqrt{(q_0 + q_1 + \lambda_1)^2 - 4 q_0 \lambda_1})$.
    Thus, 
    $$
    \LR(\beta_0) \tod q_0 + q_1 - \mu_- = \frac{1}{2} \left(q_0 + q_1 -  \lambda_1 + \sqrt{(q_0 + q_1 + \lambda_1)^2 - 4 q_0 \lambda_1}\right) \sim \Gamma(k - 1, 1, \lambda_1).
    $$
\end{proofEnd}

\begin{theoremEnd}[one big link translated={\hspace{-0.71cm} See proof on page}]{corollary}[\citeauthor{kleibergen2007generalizing}, \citeyear{kleibergen2007generalizing}]
    \label{cor:2}
    Conditionally on
    $$\lambda_1 := (n - k) \lambda_\mathrm{min}\left(\left[ \tilde X(\beta_0)^T M_Z \tilde X(\beta_0) \right]^{-1} \tilde X(\beta_0)^T P_Z \tilde X(\beta_0) \right),$$
    the random variable $\LR(\beta_0)$ is asymptotically stochastically bounded from above by
    $$\Gamma(k - m, m, \lambda_1) \overset{d}{:=} \frac{1}{2} \left(q_0 + q_1 - \lambda_1 + \sqrt{ \left(q_0 + q_1 + \lambda_1 \right)^2 - 4 q_0 \lambda_1} \right),$$
    where $q_0 \sim \chi^2(k - m)$ and $q_1 \sim \chi^2(m)$ are independent.
\end{theoremEnd}
\begin{proofEnd}
    Let $p_1(\mu)$ be equal to $p(\mu)$ but with all $\lambda_i$ replaced with $\lambda_1$:
    \begin{align*}
    p_1(\mu) &:=  (\mu - \lambda_1)^{m-1} \left( (\mu - \sum_{i=0}^m q_i) (\mu - \lambda_1) - \lambda_1 \sum_{i = 1}^m q_i.
    \right) %
    \end{align*}
    This has roots $\lambda_1$ and $\mu_\pm = \frac{1}{2} \left(\lambda_1 + \sum_{i=0}^m q_i \pm \sqrt{(\lambda_1 + \sum_{i=0}^m q_i)^2 - 4 \lambda_1 q_0} \right)$.
    The smallest root is $\mu_- < \lambda_1$.

    Define
    \begin{align*}
    g(\mu) &:= \frac{p(\mu)}{\prod_{i=1}^m (\mu - \lambda_i)} = (\mu - \sum_{i=0}^m q_i) - \sum_{i=1}^m \frac{\lambda_i q_i}{\mu - \lambda_i} \, \text{ and} \\
    g_1(\mu) &:= \frac{p_1(\mu)}{(\mu - \lambda_1)^{m}} = (\mu - \sum_{i=0}^m q_i) - \sum_{i=1}^m \frac{\lambda_1 q_i}{\mu - \lambda_1}.
    \end{align*}
    As $q_i > 0$ almost surely, for any $0 < \mu < \lambda_1 \leq \lambda_i$ we have $\frac{\lambda_i q_i}{\mu - \lambda_i} \geq \frac{\lambda_1 q_i}{\mu - \lambda_1}$ (multiply both sides by $(\mu - \lambda_1) (\mu - \lambda_i) > 0$ to verify) with equality if and only if $\lambda_i = \lambda_1$.
    Thus $g_1(\mu) \geq g(\mu)$, with equality if and only if $\lambda_i = \lambda_1$ for all $i$.
    Thus, $ g_1(\mu_\mathrm{min}) \geq g(\mu_\mathrm{min}) = 0$.
    Calculate $g_1(0) = - q_0 < 0$.
    As $g_1$ is continuous on $[0, \mu_\mathrm{min}] \subset [0, \lambda_1)$, the continuous mapping theorem implies that $g_1$ has a root in $[0, \mu_\mathrm{min}]$.
    Thus $\mu_- \leq \mu_\mathrm{min}$, with equality if and only if $\lambda_i = \lambda_1$ for all $i$.

    Thus,
    $$
    \LR(\beta_0) \tod \sum_{i=0}^m q_i - \mu_\mathrm{min} \geq \sum_{i=0}^m q_i - \mu_-
    $$
    Finally, replace $q_1 \leftarrow \sum_{i=1}^m q_i \sim \chi^2(m)$ to obtain $\LR(\beta_0) \leq \Gamma(k - m, m, \lambda_1)$.
\end{proofEnd}

\section{Computation}
To compute $p$-values based on \cref{theorem:1}, we need to approximate the cumulative distribution function of $\LR(\beta_0) \overset{d}{=} \sum_{i=0}^m q_i - \mu_\mathrm{min}$.
Using results from \citet{hillier2009conditional}, \citet{londschien2025statistician} propose to approximate the cumulative distribution function of $\Gamma(k - m, m, \lambda_1) \geq \sum_{i=0}^m q_i - \mu_\mathrm{min}$ (\cref{cor:2}) by transforming $\BP[ \Gamma(k - m, m, \lambda_1) \leq z]$ into a well-behaved one-dimensional integral.
This uses the closed-form solution for $\mu_-$, the minimal root of $p(\mu)$ for $m=1$.
For $m=2,3$, closed-form solutions for the roots of the cubic or quartic polynomial $p(\mu)$ exist, but they are not instructive.
For $m>3$, no such closed-form solutions exist.

Still, the smallest root $\mu_\mathrm{min}$ of $p(\mu)$ can be computed efficiently.
By the eigenvalue interlacing theorem, the sorted roots $\mu_i$ of $p(\mu)$ satisfy $\mu_\mathrm{min} = \mu_i \leq \lambda_1 \leq \mu_2 \leq \ldots \leq \lambda_m \leq \mu_{m+1}$ and $\mu_\mathrm{min} \leq \lambda_1$ is the only root of $p(\mu)$ in $[0, \lambda_1)$.
Define
\begin{align*}
    g(\mu) &:= \frac{p(\mu)}{\prod_{i=1}^m (\mu - \lambda_i)} = (\mu - \sum_{i=0}^m q_i) - \sum_{i=1}^m \frac{\lambda_i q_i}{\mu - \lambda_i} \, \text{ with derivative} \\
    g'(\mu) &= 1 + \sum_{i=1}^m \frac{\lambda_i q_i}{(\mu - \lambda_i)^2} > 0 \, \text{ for } \mu < \lambda_1.
\end{align*}
This is continuous and strictly increasing on $[0, \lambda_1)$ with $g(0) = -q_0 < 0$ and $\lim_{\mu \nearrow \lambda_1} g(\mu) = +\infty$.
Like $p(\mu)$, this has exactly one root in $[0, \lambda_1)$, equal to $\mu_\mathrm{min}$.
Thus, we can compute $\mu_\mathrm{min}$ by bisection or Newton's method.
We use Newton's method with a starting value of $\mu_0 = \frac{1}{2}(\sum_{i=0}^m q_i + \lambda_1 + \sqrt{ (\sum_{i=0}^m q_i + \lambda_1)^2 - 4 q_0 \lambda_1} )$, the bound from \cref{cor:2}, in the \href{https://github.com/mlondschien/ivmodels}{\texttt{ivmodels}} software package for Python.

\section{Numerical analysis}
\Cref{theorem:1} provides the exact asymptotic distribution of the likelihood-ratio test conditional on the eigenvalues of the concentration matrix.
Unless all eigenvalues are equal, the distribution is stochastically strictly smaller than the bound of \citet{kleibergen2007generalizing} (\cref{cor:2}) and using critical values based on \cref{theorem:1} leads to a strictly more powerful test.

In their analysis, \citet[][page 190]{kleibergen2007generalizing} writes that the exact distribution of $\LR(\beta_0)$ is ``indistinguishable'' from the bound $\Gamma(k - m, m, \lambda_1)$.
We observe that this assessment does not hold if the eigenvalues of the concentration matrix differ substantially.
When endogenous variables are differently identified, a common scenario in practice, using the exact distribution leads to a substantial improvement in power.

All computations were done using the \href{https://github.com/mlondschien/ivmodels}{\texttt{ivmodels}} software package for Python \citep{londschien2024weak,londschien2025statistician}.
The code to reproduce figures is available at the GitHub repository \href{https://github.com/mlondschien/ivmodels-simulations}{\texttt{github.com/mlondschien/ivmodels-simulation}}.

\subsection*{The critical value function}
\Cref{fig:1} shows the critical value functions of $\LR(\beta_0)$ at nominal level $\alpha=0.05$ according to \cref{theorem:1} under different identifications.
For $m=2, 4$ and $k = \frac{3}{2}m, \frac{5}{2}m, 5m$, we independently draw $q_0 \sim \chi^2(k-m)$ and $q_i \sim \chi^2(1)$.
We set $\lambda_1 = \Delta \lambda_1 + q_0$ and $\lambda_2 = \ldots = \lambda_m = \Delta \lambda_2 + q_0$ to avoid draws with $\mu_\mathrm{min} > \lambda_1$ as $\mu_\mathrm{min} \leq q_0$.
We compare four settings: (i)~$\Delta \lambda_1 = \Delta \lambda_2 = 5$, (ii)~$\Delta \lambda_1 = 5, \Delta \lambda_2 = 50$, (iii)~$\Delta \lambda_1 = \Delta \lambda_2 = 10$, and (iv)~$\Delta \lambda_1 = 10, \Delta \lambda_2 = 100$.
We also show the critical value function of a $\chi^2(m)$ distribution, corresponding to $\lambda_i \to \infty$ for all $i$.

\begin{figure}[htbp]
\centering
\includegraphics[width=0.9\textwidth]{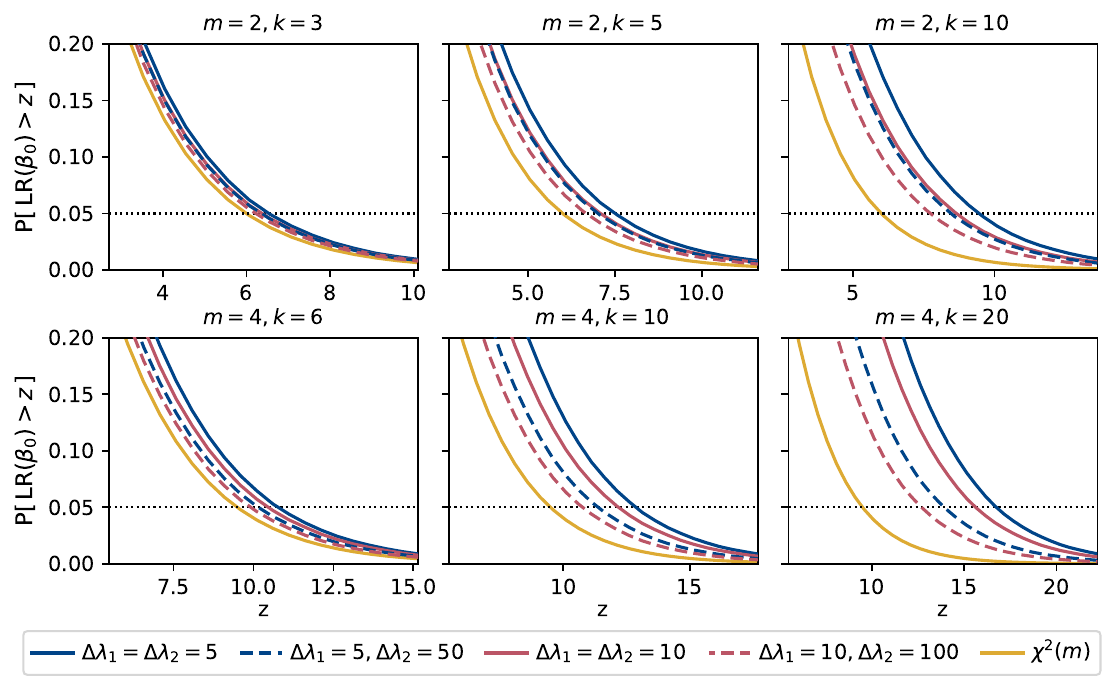}
\caption{
\label{fig:1}
Critical value functions for the conditional likelihood-ratio test conditional on $\lambda_1 = q_0 + \Delta \lambda_1$ and $\lambda_2 = \ldots = \lambda_m = q_0 + \Delta \lambda_2$.
This avoids draws with $\mu_\mathrm{min} > \lambda_1$ as $\mu_\mathrm{min} \leq q_0$.
}
\end{figure}

The critical value functions for $\Delta \lambda_1 = \Delta \lambda_2$ are exactly equal to those that would be obtained by \citeauthor{kleibergen2007generalizing}'s \citeyearpar{kleibergen2007generalizing} bound, independently of $\Delta \lambda_2$.
That is, the difference between the critical value functions for $\Delta \lambda_1 = \Delta \lambda_2$ (solid) and $10 \cdot \Delta \lambda_1 = \Delta \lambda_2$ (dashed) is exactly the increase in power achieved by using the exact distribution of \cref{theorem:1} instead of the bound of \cref{cor:2}.
For all $k, m$, the critical value function for (ii) $\Delta \lambda_1 = 5, \Delta \lambda_2 = 50$ is smaller than that for (iii) $\Delta \lambda_1 = \Delta \lambda_2 = 10$.

\subsection*{Size}
\citet{kleibergen2021efficient} shows that the asymptotic distribution of the subvector conditional likelihood-ratio test under the null depends only on $k, m$, and $\tilde \mu := n\Omega_{V \cdot \varepsilon}^{-1} \Pi^T Q \Pi$, where $\Omega_{V \cdot \varepsilon} := \Omega_V - \Omega_{V, \varepsilon} \Omega_{\varepsilon, V} / \sigma^2_\varepsilon$.
Due to rotational invariance, the asymptotic distribution of the full vector conditional likelihood-ratio then depends only on $k, m$, and the eigenvalues of $\tilde \mu$.

\Cref{fig:2} compared the empirical sizes at nominal level $\alpha = 0.05$ using \citeauthor{kleibergen2007generalizing}'s \citeyearpar{kleibergen2007generalizing} critical values (old, left) to those of \cref{theorem:1} (new, right) for $k, m = 10, 2$ (top) and $k, m = 20, 4$ (bottom).
We draw $n=1000$ samples from a Gaussian linear model with $\tilde \mu = \diag(\lambda_1, \lambda_2)$ ($m=2$, top) and $\tilde \mu = \diag(\lambda_1, \lambda_2, \lambda_2, \lambda_2)$ ($m=4$, bottom) for $\lambda_1, \lambda_2 = 1, \ldots, 100$ and show the proportion of rejections out of $50'000$ simulations for each grid point.

The empirical size of the conditional likelihood-ratio test using  \citeauthor{kleibergen2007generalizing}'s \citeyearpar{kleibergen2007generalizing} critical values varies with $\lambda_1, \lambda_2$ and drops substantially below the nominal level $\alpha = 0.05$ if $\lambda_1$ and $\lambda_2$ are of a different magnitude.
In contrast, up to noise, the empirical size of the conditional likelihood-ratio test using the critical values of \cref{theorem:1} is constant and equal to the nominal level $\alpha = 0.05$.

Note that $\tilde \mu = \diag(\lambda_1, \ldots, \lambda_4)$ does not imply that $\lambda_1, \ldots, \lambda_4$ are the eigenvalues of the empirical version of the concentration matrix $(n - k) [\tilde X(\beta_0)^T M_Z \tilde X(\beta_0) ]^{-1} \tilde X(\beta_0)^T P_Z \tilde X(\beta_0)$.
This explains why the rejection rates using \citeauthor{kleibergen2007generalizing}'s \citeyearpar{kleibergen2007generalizing} critical values (left) are not exactly equal to the nominal level $\alpha = 0.05$ on the diagonal $\lambda_1 = \lambda_2$.

\begin{figure}[htbp]
\centering
\includegraphics[width=\textwidth]{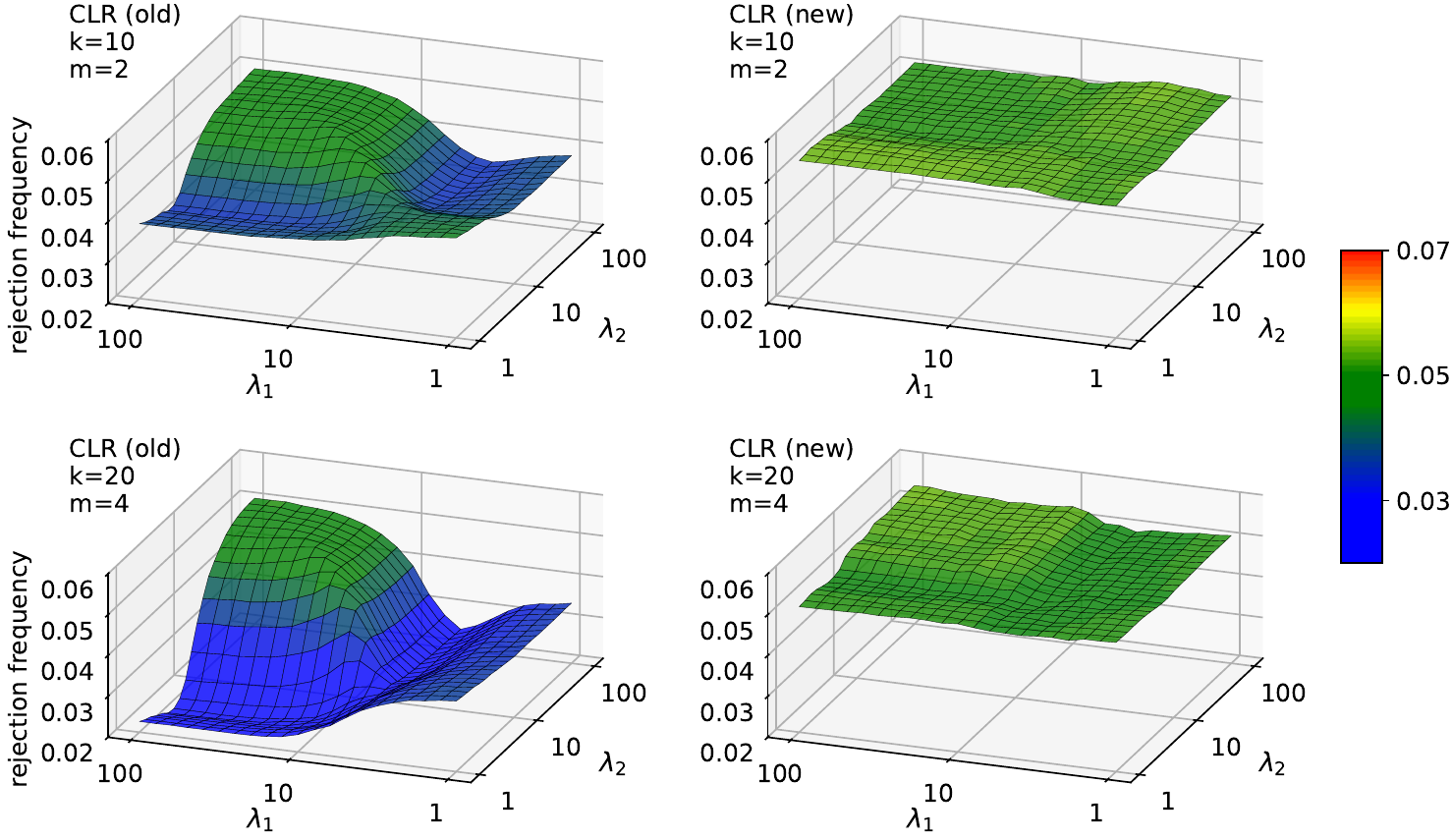}
\caption{
\label{fig:2}
Empirical sizes of the conditional likelihood-ratio test with \citeauthor{kleibergen2007generalizing}'s \citeyearpar{kleibergen2007generalizing} critical values (old, left) and the exact critical values from \cref{theorem:1} (new, right) at the nominal level $\alpha=0.05$.
We draw data from a Gaussian linear model with concentration matrix $n \Omega_{V \cdot \varepsilon}^{-1} \Pi^T Q \Pi = \mathrm{diag}(\lambda_1, \lambda_2)$ ($k=10$, $m=2$, top) and $n \Omega_{V \cdot \varepsilon}^{-1} \Pi^T Q \Pi = \mathrm{diag}(\lambda_1, \lambda_2, \lambda_2, \lambda_2)$ ($k=20, m=4$, bottom), varying $\lambda_1, \lambda_2 = 1, \ldots, 100$ over a logarithmic grid of $21 \times 21$ and show empirical rejection rates over $50'000$ draws for each grid point.
}
\end{figure}

\subsection*{Power}
Finally, we numerically analyse the power difference of the conditional likelihood-ratio test at nominal level $\alpha = 0.05$ using \citeauthor{kleibergen2007generalizing}'s \citeyearpar{kleibergen2007generalizing} critical values and those of \cref{theorem:1}.
For $i=1, \ldots, 1000$, we independently draw $Z_i \sim \mathcal{N}(0, \Id_k)$, $\Pi \in \BR^{k \times m}$ such that $n \Pi^T \Pi = \mathrm{diag}(\lambda_1, \lambda_2, \ldots, \lambda_2)$, and $V_{X_i} \sim \mathcal{N}(0, \Id_m)$ and $y_i = \varepsilon_i \sim \mathcal{N}(0, 1)$ (that is, $\beta_0 = 0$) jointly Gaussian with $\Cov(V_{X, i}, \varepsilon_i) = (-0.5, 0, \ldots, 0)$.
We fix $\lambda_1 = 5$ (left), $10$ (right) and vary $\lambda_2 = 1, \ldots, 100$ and $\beta = \beta_1 \cdot e_1$ for $\beta_1 = -1, \ldots, 1$.
In \cref{fig:3}, we show difference between the empirical rejection rates at nominal level $\alpha =0.05$ using \citeauthor{kleibergen2007generalizing}'s \citeyearpar{kleibergen2007generalizing} critical values and those of \cref{theorem:1}.
We observe that the critical values of \cref{theorem:1} result in a substantially more powerful test, with a difference in rejection rates at level $\alpha = 0.05$ of up to 6\% ($k=10, m=2$, top) and up to 14\% ($k=20, m=4$, bottom).
\begin{figure}[htbp]
\centering
\includegraphics[width=\textwidth]{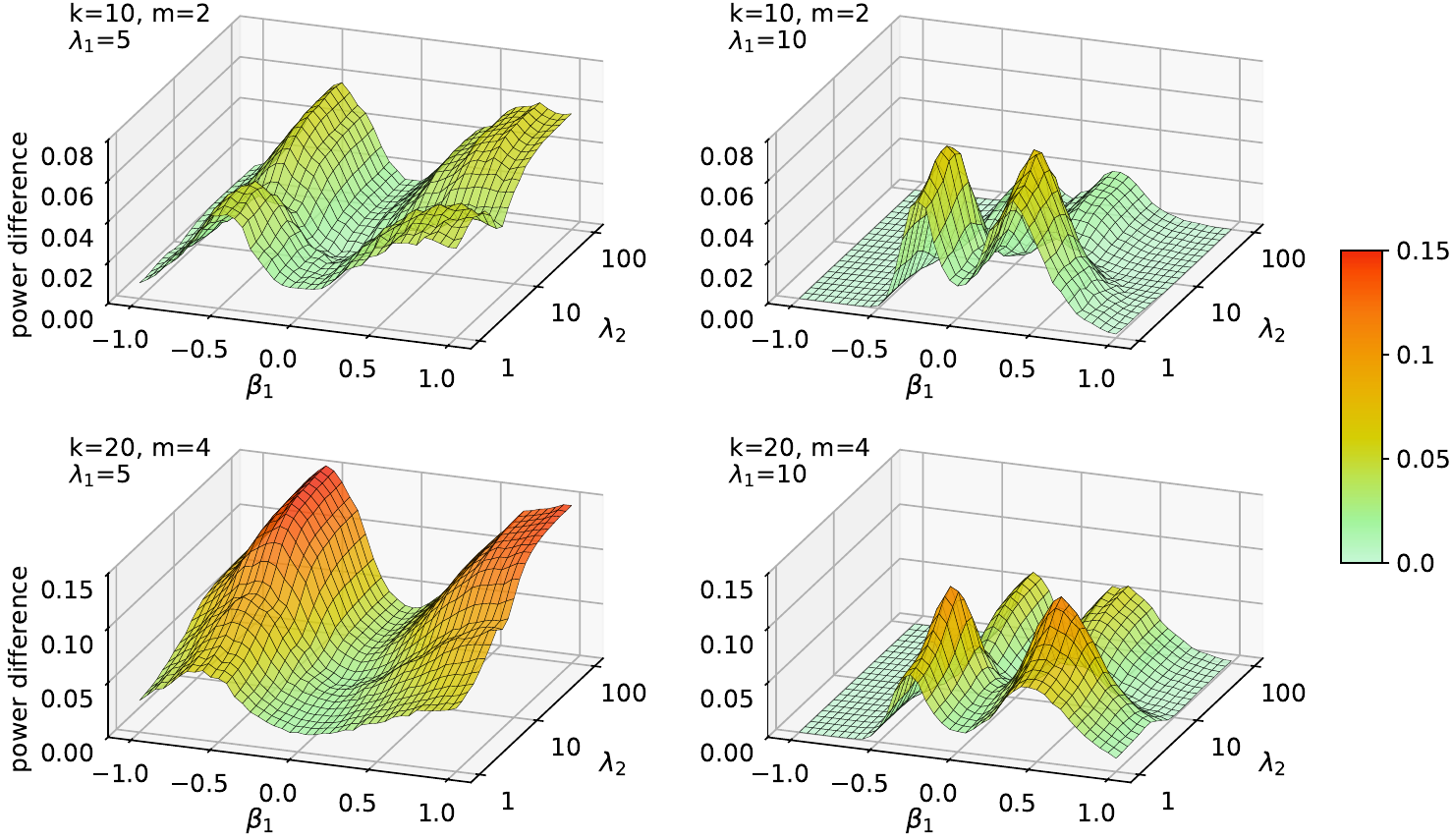}
\caption{
\label{fig:3}
Power difference of the conditional likelihood-ratio test at the significance level ${\alpha=0.05}$ using \citeauthor{kleibergen2007generalizing}'s \citeyearpar{kleibergen2007generalizing} critical values and those of \cref{theorem:1}.
We vary $\beta = \beta_1 \cdot e_1$ for $\beta_1 = -1, \ldots, 1$ linearly spaced with 41 values and vary $\lambda_2 = 1, \ldots, 100$, the identification of the variables other than $X_1$, logarithmically spaced with 21 values.
The concentration matrix is $n\Omega_V^{-1} \Pi^T Q \Pi = \mathrm{diag}(\lambda_1, \lambda_2)$ ($k=10$, $m=2$, top) and $n\Omega_V^{-1} \Pi^T Q \Pi = \mathrm{diag}(\lambda_1, \lambda_2, \lambda_2, \lambda_2)$ ($k=20, m=4$, bottom) for $\lambda_1 = 5$ (left) and $\lambda_1 = 10$ (right).
The power difference is computed over $20'000$ simulations for each grid point.
}
\end{figure}

\FloatBarrier

\bibliography{bib}

\appendix
\section{Proofs}

\begin{lemma}
    \label{lem:1}
    The arrowhead matrix
    $$
    A =
    \begin{pmatrix}
        d_0 & a_1 & a_2 & \cdots & a_l \\
        a_1 & d_1 & 0 & \cdots & 0 \\
        a_2 & 0 & d_2 & \cdots & 0 \\
        \vdots & \vdots & \vdots & \ddots & \vdots \\
        a_l & 0 & 0 & \cdots & d_l
    \end{pmatrix}
    $$
    has determinant
    $$
    \det(A) = \prod_{i = 0}^l d_i - \sum_{i = 1}^l \prod_{j \geq 1, j \neq i} d_j \cdot a_i^2
    $$
\end{lemma}
\begin{proof}
Assume that $d_1, \ldots, d_l \neq 0$.
We remove the non-zero entries $a_i$ in the first row by Gauss elimination: For $i = 1, \ldots l$, we substract the $i+1$-th row times $\frac{a_i}{d_i}$ from the first row.
This preserves the determinant.
Thus,
\begin{align*}
\det(A) &= \det \begin{pmatrix}
    d_0 - \sum_{i = 1}^l \frac{a_i^2}{d_i} & 0 & 0 & \cdots & 0 \\
    a_1 & d_1 & 0 & \cdots & 0 \\
    a_2 & 0 & d_2 & \cdots & 0 \\
    \vdots & \vdots & \vdots & \ddots & \vdots \\
    a_l & 0 & 0 & \cdots & d_l
\end{pmatrix} \\
\numberthis
\label{eq:det_arrowhead}
&= \left( d_0 - \sum_{i = 1}^l \frac{a_i^2}{d_i} \right) \cdot \prod_{i = 1}^l d_i = \prod_{i = 1}^l d_i - \sum_{i = 1}^l \prod_{j \geq 1, j \neq i} d_j \cdot a_i^2.
\end{align*}
This polynomial is continuous in $d_1, \ldots, d_l$, as is the determinant of $A$ in its entries.
Equation \eqref{eq:det_arrowhead} thus holds for all $d_1, \ldots, d_l$ by continuity.
\end{proof}

\printProofs
\end{document}